\documentclass{article}

\usepackage{hyperref}
\usepackage{verbatim}
\usepackage{amsmath, amsthm, amssymb, graphics, mathtools, graphicx, xcolor}
\usepackage{enumitem}

\newtheorem{theorem}{Theorem}[section]
\newtheorem{lemma}[theorem]{Lemma}

\newtheorem{claim}{Claim}

\makeatletter
\newtheorem*{rep@theorem}{\rep@title}
\newcommand{\newreptheorem}[2]{%
\newenvironment{rep#1}[1]{%
 \def\rep@title{#2 \ref{##1}}%
 \begin{rep@theorem}}%
 {\end{rep@theorem}}}
\makeatother

\newreptheorem{theorem}{Theorem}
\newreptheorem{corollary}{Corollary}
\newreptheorem{lemma}{Lemma}

\newcommand{\bpc}{\noindent {\em Proof of Claim~\theclaim. }}
\newcommand{\epc}{This proves Claim~\theclaim.}

\newcommand{\sm}{\setminus}

\newcommand{\seq}{\subseteq}

\newcommand{\cc}{\textsc{cc}}

\newenvironment{myproof}[1][\proofname]{%
  \begin{proof}[#1]$ $\par\nobreak\ignorespaces
}{%
  \end{proof}
}

\begin{document}

\title{(Theta, triangle)-free and (even hole, $K_4$)-free graphs. Part 2: Bounds on treewidth}

\author{Marcin Pilipczuk\thanks{Institute of Informatics, University
    of Warsaw Banacha 2, 02-097 Warsaw, Poland\newline This research is a
    part of a project that has received funding from the European
    Research Council (ERC) under the European Union’s Horizon 2020
    research and innovation programme Grant Agreement no. 714704.}~,
  Ni Luh Dewi Sintiari\thanks{Univ Lyon, EnsL, UCBL, CNRS, LIP,
    F-69342, LYON Cedex 07, France. \newline The last three authors
    are partially supported by the LABEX MILYON (ANR-10-LABX-0070) of
    Universit\'e de Lyon, within the program ‘‘Investissements
    d'Avenir’’ (ANR-11-IDEX-0007) operated by the French National
    Research Agency (ANR) and by Agence Nationale de la Recherche (France) under 
    research grant ANR DIGRAPHS ANR-19-CE48-0013-01.}~,\\St\'ephan Thomass\'e\footnotemark[2]~
  and Nicolas Trotignon\footnotemark[2]}


\maketitle
\begin{abstract}
  A {\em theta} is a graph made of three internally vertex-disjoint
  chordless paths $P_1 = a \dots b$, $P_2 = a \dots b$,
  $P_3 = a \dots b$ of length at least~2 and such that no edges exist
  between the paths except the three edges incident to $a$ and the
  three edges incident to $b$.  A {\em pyramid} is a graph made of
  three chordless paths $P_1 = a \dots b_1$, $P_2 = a \dots b_2$,
  $P_3 = a \dots b_3$ of length at least~1, two of which have length
  at least 2, vertex-disjoint except at $a$, and such that $b_1b_2b_3$
  is a triangle and no edges exist between the paths except those of
  the triangle and the three edges incident to~$a$.  An \emph{even
    hole} is a chordless cycle of even length.  For three non-negative
  integers $i\leq j\leq k$, let $S_{i,j,k}$ be the tree with a vertex
  $v$, from which start three paths with $i$, $j$, and $k$ edges
  respectively.  We denote by $K_t$ the complete graph on $t$
  vertices.

  We prove that for all non-negative integers $i, j, k$, the class of
  graphs that contain no theta, no $K_3$, and no $S_{i, j, k}$ as
  induced subgraphs have bounded treewidth. We prove that for all
  non-negative integers $i, j, k, t$, the class of graphs that contain
  no even hole, no pyramid, no $K_t$, and no $S_{i, j, k}$ as induced
  subgraphs have bounded treewidth. To bound the treewidth, we prove
  that every graph of large treewidth must contain a large clique or a
  minimal separator of large cardinality.
\end{abstract}

\section{Introduction}
\label{sec:introduction}

In this article, all graphs are finite, simple, and undirected.  A
graph $H$ is an {\em induced subgraph} of a graph $G$ if some graph
isomorphic to $H$ can be obtained from $G$ by deleting vertices.  A
graph $G$ {\em contains $H$} if $H$ is an induced subgraph of $G$.  A
graph is {\em $H$-free} if it does not contain $H$. For a family of
graphs ${\mathcal H}$, $G$ is {\em ${\mathcal H}$-free} if for every
$H\in {\mathcal H}$, $G$ is $H$-free.  

A {\em hole} in a graph is a chordless cycle of length at least~4. It
is \emph{odd} or \emph{even} according to its length (that is its
number of edges). We denote by $K_t$ the complete graph on $t$
vertices.

\begin{figure}
  \label{fig:tc}
  \begin{center}
    \includegraphics[height=2cm]{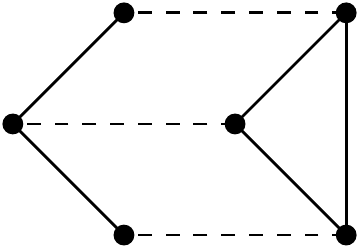}
    \hspace{.2em}
    \includegraphics[height=2cm]{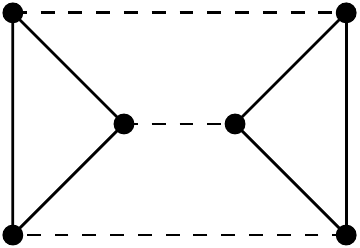}
    \hspace{.2em}
    \includegraphics[height=2cm]{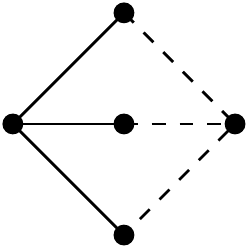}
    \hspace{.2em}
    \includegraphics[height=2cm]{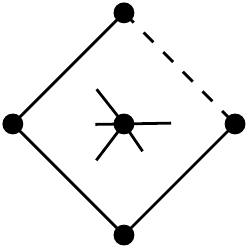}
  \end{center}
  \caption{Pyramid, prism, theta, and wheel (dashed lines represent
    paths)}
\end{figure}

A {\em theta} is a graph made of three internally vertex-disjoint
chordless paths $P_1 = a \dots b$, $P_2 = a \dots b$,
$P_3 = a \dots b$ of length at least~2 and such that no edges exist
between the paths except the three edges incident to $a$ and the three
edges incident to $b$ (see Fig.~\ref{fig:tc}).  Observe that a theta
contains an even hole, because at least two paths in the theta have
lengths of same parity and therefore induce an even hole.

We are interested in understanding the structure of even-hole-free
graphs and theta-free graphs.  Our motivation for this is explained in
the first paper of this
series~\cite{DBLP:journals/corr/abs-1906-10998}, where we give a
construction that we call layered wheel, showing that the cliquewith,
the rankwidth, and the treewidth of (theta, triangle)-free and (even
hole, $K_4$)-free graphs are unbounded. We also indicate questions,
suggested by this construction, about the induced subgraphs contained
in graphs with large treewidth.

In this second and last part, we prove that when excluding more
induced subgraphs, there is an upper bound on the treewidth.  Our results imply that the
maximum independent set problem can be solved in polynomial time for
some classes of graphs that are possibly of interest because they are
related to several well known open questions in the field.

\subsection*{Results}

We denote by $P_k$ the path on $k$ vertices. 
For three non-negative integers $i\leq j\leq k$, let $S_{i,j,k}$ be
the tree with a vertex $v$, from which start three paths with $i$,
$j$, and $k$ edges respectively.  Note that $S_{0, 0, k}$ is a path of
length $k$ (so, is equivalent to $P_{k+1}$) and that $S_{0, i, j} = S_{0, 0, i+j}$.  The \emph{claw} is
the graph $S_{1,1,1}$. Note that $\{S_{i,j,k}; 1\leq i\leq j \leq k\}$
is the set of all the subdivided claws and
$\{S_{i,j,k}; 0\leq i\leq j \leq k\}$ is the set of all subdivided
claws and paths.

A {\em pyramid} is a graph made of three chordless paths
$P_1 = a \dots b_1$, $P_2 = a \dots b_2$, $P_3 = a \dots b_3$ of
length at least~1, two of which have length at least 2, vertex-disjoint
except at $a$, and such that $b_1b_2b_3$ is a triangle and no edges
exist between the paths except those of the triangle and the three
edges incident to~$a$ (see Fig.~\ref{fig:tc}).

We do not not recall here the definition of treewidth and cliquewidth.
They are parameters that measure how complex a graph
is. See~\cite{DBLP:journals/jgt/HarveyW17,DBLP:journals/corr/abs-1901-00335}
for surveys about them.

Our main result states that for every fixed non-negative integers $i, j, k ,t$,
the following graph classes have bounded treewidth:

\begin{itemize}
\item (theta, triangle, $S_{i, j, k}$)-free graphs; 
\item (even hole, pyramid, $K_t$, $S_{i, j, k}$)-free graphs. 
\end{itemize}

The exact bounds and the proofs are given in Section~\ref{sec:bounding-tw}
(Theorems~\ref{th:ttf} and~\ref{th:ehf}).  In fact, the class on which
we actually work is larger. It is a common generalization $\cal C$ of
the graphs that we have to handle in the proofs for the two bounds
above.  Also, we do not exclude $S_{i, j, k}$, but some graphs that
contain it, the so-called $l$-span-wheels for sufficiently large
$l$. We postpone the definitions of $\cal C$ and of span wheels to
Section~\ref{sec:bounding-tw}.

To bound the treewidth, we prove that every
graph of large treewidth must contain a large clique or a minimal
separator of large cardinality.  Let us define them. 

For two vertices $s, t \in V(G)$, a set $X \subseteq V(G)$ is an {\em
  $st$-separator} if $s,t \notin X$ and $s$ and $t$ lie in different
connected components of $G\sm X$.  An $st$-separator $X$ is a {\em
  minimal $st$-separator} if it is an inclusion-wise minimal
$st$-separator. A set $X \subseteq V(G)$ is a {\em separator} if there
exist $s,t \in V(G)$ such that $X$ is an $st$-separator in $G$.  A set
$X \subseteq V(G)$ is a {\em minimal separator} if there exist
$s,t \in V(G)$ such that $X$ is a minimal $st$-separator in $G$.

 Our graphs have no large cliques by definition, and by studying their
 structure, we prove that they cannot contain large minimal
 separators, implying that their treewidth is bounded.

 Note that from the celebrated grid-minor theorem, it is easy to see
 that every graph of large treewidth contains a subgraph with a large
 minimal separator (a column in the middle of the grid contains such a
 separator). But since we are interested in the induced subgraph
 containment relation, we cannot delete edges and we have to rely on
 our reinforcement.

\subsection*{Treewidth and cliquewidth of some classes of graphs}

We now survey results about the treewidth in classes of graphs related
to the present work.  Complete graphs provide trivial examples of
even-hole-free graphs of arbitrarily large treewidth.
In~\cite{DBLP:journals/dm/CameronSHV18}, it is proved that (even hole,
triangle)-free graphs have bounded treewidth (this is based on a
structural description from~\cite{DBLP:journals/jgt/ConfortiCKV00}).
In~\cite{DBLP:journals/jgt/CameronCH18}, it is proved that for every
positive integer $t$, (even hole, pan, $K_t$)-free graphs have bounded
treewidth (where a \emph{pan} is any graph that consists of a hole and
a single vertex with precisely one neighbor on the hole).  It is
proved in~\cite{DBLP:journals/corr/abs-1906-10998} that the treewidth
of (theta, triangle)-free graphs and (even hole, pyramid, $K_4$)-free
graphs are unbounded.  Growing the treewidth 
in~\cite{DBLP:journals/corr/abs-1906-10998} requires introducing in
the graph a large clique minor and vertices of large degree.  It is
therefore natural to ask whether these two conditions are really
needed, and the answer is yes for both of them, because
in~\cite{adler} it is proved that even-hole-free graphs with no
$K_t$-minor have bounded treewidth, and in~\cite{tara} it is proved
that even-hole-free graphs with maximum degree~$t$ have bounded
treewidth.

Since having bounded cliquewidth is a weaker property than having
bounded treewidth but still has nice algorithmic applications, we
survey some results about the cliquewidth in classes related to the
present work.

It is proved in~\cite{DBLP:journals/dm/CameronSHV18} that (even hole,
cap)-free graphs with no clique separator have bounded cliquewidth
(where a \emph{cap} is any graph that consists of a hole and a single
vertex with precisely two adjacent neighbors on the hole, and a
\emph{clique separator} is a separator that is a clique).  It is
proved in~\cite{DBLP:journals/corr/abs-1906-10998}, that (triangle,
theta)-free and (even hole, pyramid, $K_4$)-free graphs have unbounded
cliquewidth.  It is proved in~\cite{adlerLMRTV:rwehf}, that (even
hole, diamond)-free graphs with no clique separator have unbounded
cliquewidth (the \emph{diamond} is the graph obtained from $K_4$ by
deleting an edge). The construction can be easily extended to (even
hole, pyramid, diamond)-free graphs as explained
in~\cite{chudetal:maxStEHFPyramfree}.  It is easy to provide (theta,
$K_4$, $S_{1, 1, 1}$)-free graphs (or equivalently (claw, $K_4$)-free
graphs) of unbounded cliquewidth.  To do so, consider a \emph{wall}
$W$, subdivide all edges to obtain $W'$, and take the line graph
$L(W')$ (see~\cite{DBLP:journals/corr/abs-1906-10998} for a definition
and Fig.~\ref{fig:line}).

\begin{figure}
  \begin{center}
    \includegraphics[height=2.5cm]{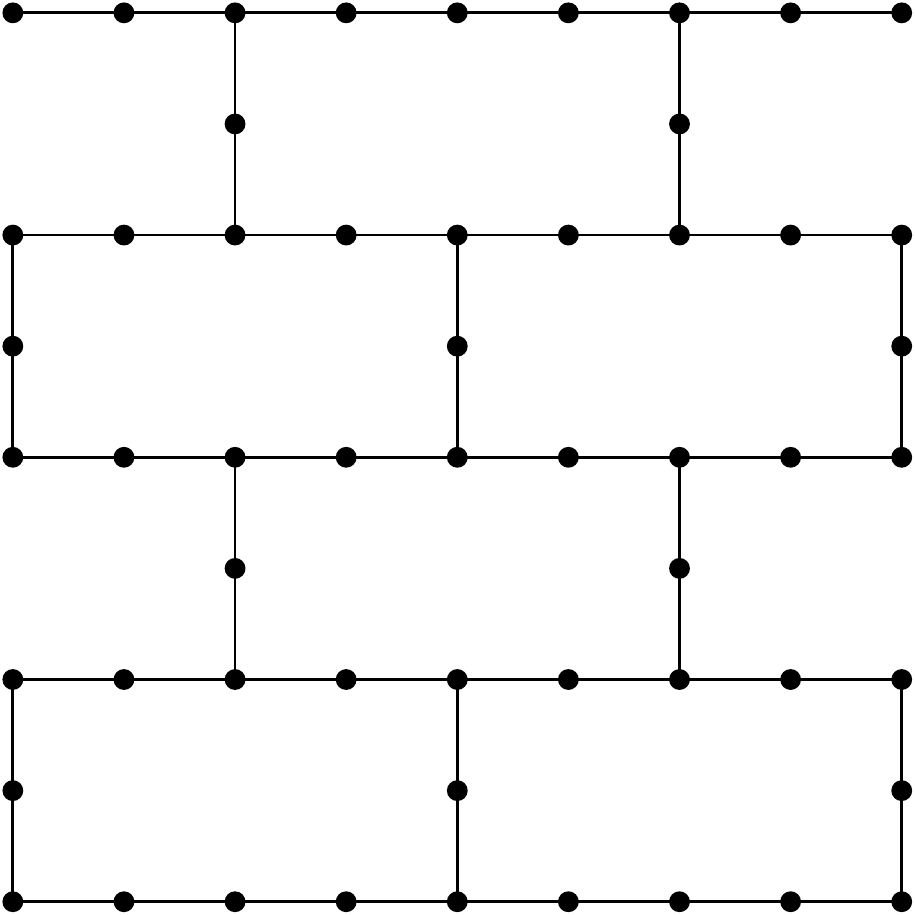}
    \hspace{.2em}
    \includegraphics[height=2.5cm]{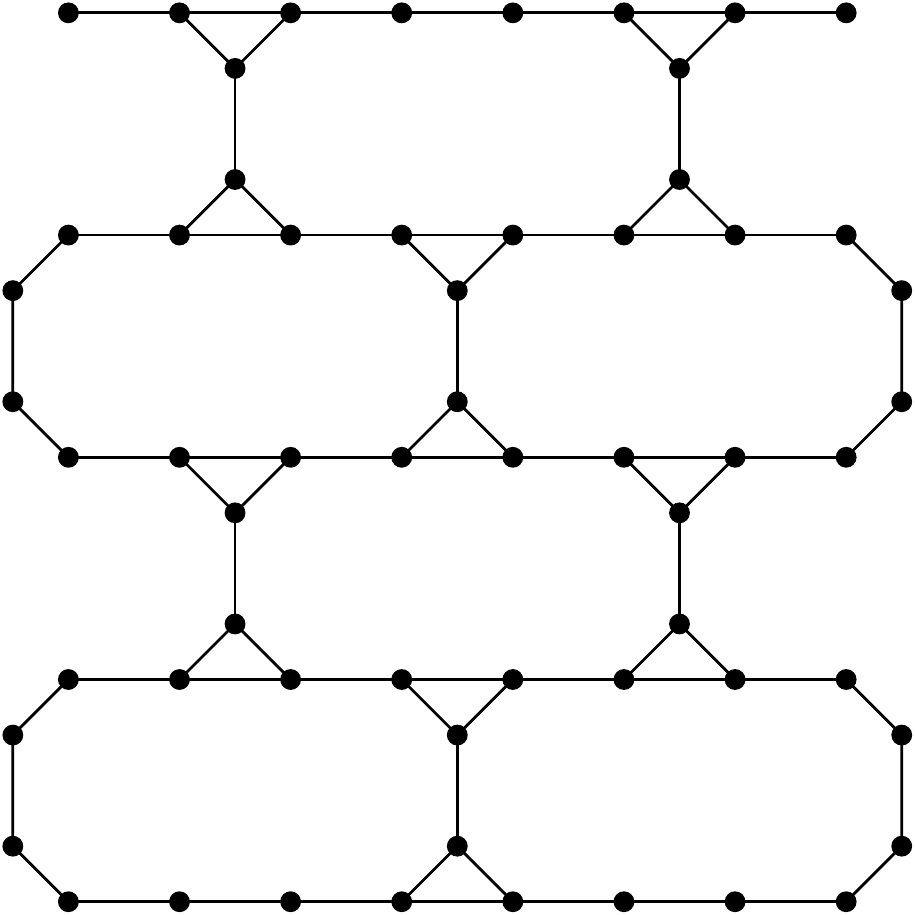}
  \end{center}
  \caption{A subdivision of a wall and its line graph\label{fig:line}}
\end{figure}

The results mentioned in this paragraph are extracted
from~\cite{DBLP:journals/cj/DabrowskiP16} (but some of them were first
proved in other works).  Let
$\mathcal H_U = \{P_7, S_{1, 1, 4}, S_{2, 2, 2}\}$ and
$\mathcal H_B = \{P_6, S_{1, 1, 3}\}$.  If $H$ contains a graph from
$\mathcal H_U$ as an induced subgraph, then the class of (triangle,
$H$)-free graph has unbounded cliquewidth (see Theorem 7.ii.6
in~\cite{DBLP:journals/cj/DabrowskiP16}). If $H$ is contained in a
graph from $\mathcal H_B$, then the class of (triangle, $H$)-free
graphs has bounded cliquewidth (see Theorem 7.i.3
in~\cite{DBLP:journals/cj/DabrowskiP16}).

The cliquewidth of (triangle, $S_{1, 2, 2}$)-free graphs is bounded,
see~\cite{DBLP:journals/corr/BrandstadtMM16}
or~\cite{DBLP:journals/jcss/DabrowskiDP17}.

\subsection*{Algorithmic consequences}

It is proved in~\cite{DBLP:journals/iandc/Courcelle90} that in every
class of graphs of bounded treewidth, many problems can be solved in
polynomial time. Our result has therefore applications to several
problems, but we here focus on one because the induced subgraphs that
are excluded in the most classical results and open questions about it
seem to be related to our classes.

An \emph{independent set} in a graph is a set of pairwise non-adjacent
vertices.  Our results imply that computing an independent set of
maximum cardinality can be performed in polynomial time for (theta,
triangle, $S_{i, j, k}$)-free graphs and (even hole, pyramid, $K_t$,
$S_{i, j, k}$)-free graphs.

Finding an independent set of maximum cardinality is polynomial time
solvable for (even hole, triangle)-free
graphs~\cite{DBLP:journals/dm/CameronSHV18} and (even hole,
pyramid)-free graphs~\cite{chudetal:maxStEHFPyramfree}.  Its
complexity is not known for (even hole, $K_4$)-free graphs and for
(theta, triangle)-free graphs.  Determining its complexity is also a
well known question for $S_{i, j, k}$-free graphs. It is NP-hard for
the class of $H$-free graphs whenever $H$ is not an induced subgraph
of some $S_{i, j, k}$~\cite{alekseev:83}.  It is solvable in
polynomial time for $H$-free graphs whenever $H$ is contained
in $P_k$ for $k= 6$ (see~\cite{DBLP:conf/soda/LokshantovVV14} for
$H=P_5$ and~\cite{DBLP:conf/soda/GrzesikKPP19} for $H=P_6$) or
contained in $S_{i, j, k}$ with $(i, j, k) \leq (1, 1, 2)$
(see~\cite{DBLP:journals/dam/Alekseev04}
and~\cite{DBLP:journals/jda/LozinM08} for the weighted version).  It
is solvable in polynomial time for ($P_7$, triangle)-free
graphs~\cite{DBLP:journals/dam/BrandstadtM18} and for ($S_{1,2, 4}$,
triangle)-free graphs~\cite{DBLP:journals/corr/abs-1806-09472}.  The
complexity is not known for $H$-free graphs whenever $H$ is some
$S_{i, j, k}$ that contains either $P_7$, $S_{1, 1, 3}$, or
$S_{1, 2, 2}$.

\subsection*{Bounding the number of minimal separators}

One possible method to find maximum weight independent sets for a
class of graphs is by proving that every graph in the class has
polynomially many minimal separators (where the polynomial is in the
number of vertices of the graph). This was for instance successfully
applied to (even hole, pyramid)-free graphs
in~\cite{chudetal:maxStEHFPyramfree}. Therefore, our result on (even
hole, pyramid, $K_t$, $S_{i, j, k}$)-free graphs does not settle a new
complexity result for the Maximum Independent Set problem (but it
still can be applied to other problems).

Note that bounding the number of minimal separators cannot be applied
to (even hole, $K_4$)-free graphs and to (theta, triangle)-free graphs
since there exist graphs in both classes that contain exponentially
many minimal separators.  These graphs are called \emph{$k$-turtle} and
\emph{$k$-ladder}, see Fig~\ref{fig:examples}.  It is
straightforward to check that they have exponentially many minimal
separators (the idea is that a separator can be built by making a
choice in each horizontal edge, and there are $k$ of them).  Moreover,
$k$-turtles are (theta, triangle)-free (provided that the outer cycle
is sufficiently subdivided) and $k$-ladders are (even hole,
$K_4$)-free.

\begin{figure}
  \begin{center}
    \includegraphics[height=4cm]{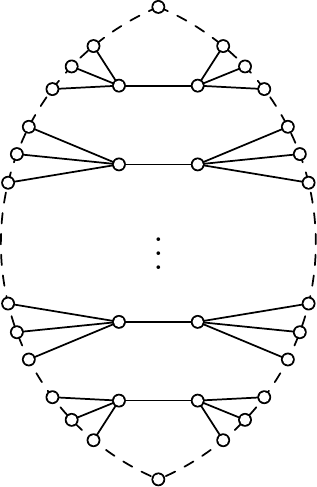}
    \hspace{2em}
    \includegraphics[height=4cm]{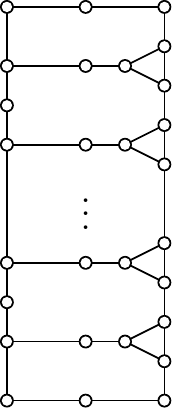}
  \end{center}
  \caption{$k$-turtle and $k$-ladder (dashed lines represent
    paths)\label{fig:examples}}
\end{figure}

\subsection*{Open questions}

It is not known whether (even hole, $K_4$, diamond)-free graphs have
bounded treewidth (or cliquewidth). Also, for every fixed integer
$t\geq 4$, it is not known whether (theta, triangle)-free graphs of
maximum degree $t$ have bounded treewidth (for $t=1, 2$, the treewidth
is trivially bounded and for $t=3$ it follows from Corollary~4.3 in
\cite{adler}).  It is not known whether (triangle, $S_{1, 2, 3}$)-free
graphs have bounded cliquewidth,
see~\cite{DBLP:journals/jcss/DabrowskiDP17} for other open problems of
the same flavor.

\subsection*{Outline of the paper}

In Section~\ref{sec:pilip}, we explain our method to bound the
treewidth. In Section~\ref{sec:technical-lemma}, we give two technical
lemmas that highlight structural similarities between (theta,
triangle)-free and (even hole, pyramid)-free graphs.  These will be
used in Section~\ref{sec:bounding-tw} where we prove that graphs in
our classes do not contain minimal separators of large cardinality,
implying that their treewidth is bounded.

\subsection*{Notation}

By \emph{path} we mean chordless (or induced) path. When $a$ and $b$
are vertices of a path $P$, we denote by $aPb$ the subpath of $P$ with
ends $a$ and $b$.

When $A, B\subseteq V(G)$, we denote by $N_B(A)$ the set of vertices
from $B\sm A$ that have at least one neighbor in $A$ and $N(A)$ means
$N_{V(G)}(A)$.  Note that $N_B(A)$ is disjoint from $A$.  We write
$N(a)$ instead of $N(\{a\})$ and $N[a]$ for $\{a\}\cup N(a)$.  We
denote by $G[A]$ the subgraph of $G$ induced by $A$.  To avoid too
heavy notation, since there is no risk of confusion, when $H$ is an
induced subgraph of $G$, we write $N_H$ instead of $N_{V(H)}$.

A vertex $x$ is \emph{complete} (resp.\ \emph{anticomplete}) to $A$ if
$x\notin A$ and $x$ is adjacent to all vertices of $A$ (resp.\ to no
vertex of $A$). We say that $A$ is \emph{complete} (resp.\
\emph{anticomplete}) to $B$ if every vertex of $A$ is complete
(resp. anticomplete) to $B$ (note that this means in particular that
$A$ and $B$ are disjoint).

 \section{Treewidth and minimal separators}
 \label{sec:pilip}

 If a graph has large treewidth, then it contains some sub-structure
 that is highly connected in some sense (grid minor, bramble, tangle,
 see \cite{DBLP:journals/jgt/HarveyW17}).
 Theorem~\ref{th:PMC-Pilipczuk-thm} seems to be a new statement of
 that kind.  It says that graphs of large treewidth must contain
 either a large clique or a minimal separator of large size.  However,
 its converse is false, as shown by $K_{2, t}$ that has treewidth~2
 (it is a series-parallel graph) and contains a minimal separator of
 size~$t$.

 A variant of the following theorem can be obtained from the
 celebrated excluded grid theorem of Robertson and Seymour.  The idea
 is to use a large grid to obtain a large minimal separator.  But
 there are technicalities because we are not allowed to delete edges,
 so the grid might contain many crossing edges. To find two vertices
 that cannot be separated by a small separator, one needs to clean the
 grid.  We do not include the details since
 the following provides a better bound. 

\begin{theorem} 
\label{th:PMC-Pilipczuk-thm}
Let $G$ be a graph and let $k \geq 2$ and $s \geq 1$ be positive
integers. If $G$ does not contain a clique on $k$ vertices nor a
minimal separator of size larger than $s$, then the treewidth of $G$
is at most $(k-1)s^3 - 1$.
\end{theorem}

Before proving Theorem~\ref{th:PMC-Pilipczuk-thm}, let us introduce
some terminology and state results due to Bouchitt\'e and
Todinca~\cite{DBLP:journals/siamcomp/BouchitteT01}.  For a graph $G$
we denote by $\cc(G)$ the set of all connected components of $G$
(viewed as subsets of $V(G)$).  A set
$F \seq {V(G) \choose 2} \sm E(G)$ is a {\em fill-in} or {\em chordal
  completion} if $G + F = (V(G), E(G) \cup F)$ is a chordal graph.  A
fill-in $F$ is {\em minimal} if it is inclusion-wise minimal.  If
$X \subseteq V(G)$, then every connected component
$D \in \cc(G \sm X)$ with $N(D) = X$ is called a component {\em full}
to $X$. Observe that a set $X \subseteq V(G)$ is a minimal separator
if and only if there exist at least two connected components of $G\sm X$
that are full to $X$.  An important property of minimal separators is
that no new minimal separator appears when applying a minimal fill-in.

\begin{lemma}[see \cite{DBLP:journals/siamcomp/BouchitteT01}]
  For every graph $G$, minimal fill-in $F$, and  minimal separator
  $X$ in $G + F$, $X$ is a minimal separator in $G$ as
  well. Furthermore, the families of components $\cc((G + F) \sm X)$
  and $\cc(G \sm X)$ are equal (as families of subsets of $V(G)$).
\end{lemma}

A set $\Omega \subseteq V(G)$ is a {\em potential maximal clique (PMC)} 
if there exists a minimal fill-in $F$ such that $\Omega$ is a
maximal clique of $G + F$. A PMC is surrounded by minimal separators.

\begin{lemma}[see \cite{DBLP:journals/siamcomp/BouchitteT01}]
\label{lem:PMC1}
  For every PMC $\Omega$ in $G$ and every component $D \in \cc(G \sm \Omega)$, the set $N(D)$ is a minimal
  separator in $G$ with $D$ being a full component.
\end{lemma}

The following characterizes PMCs.

\begin{theorem}[see \cite{DBLP:journals/siamcomp/BouchitteT01}]
\label{th:PMC2}
  A set $\Omega \subseteq V(G)$ is a PMC in $G$ if and only if the following two conditions hold:
  \begin{enumerate}[label=(\roman*)]
    \item for every $D \in \cc(G \sm \Omega)$ we have $N(D) \subsetneq \Omega$;
    \item for every $x, y \in \Omega$ either $x = y$, $xy \in E(G)$, or there exists $D \in \cc(G \sm \Omega)$ with $x, y \in N(D)$.
  \end{enumerate}
\end{theorem}

In the second condition of Theorem~\ref{th:PMC2}, we say that a component $D$ {\em covers} the nonedge $xy$.

\newcommand{\pmc}{\Omega}

\begin{lemma} 
\label{lem:PMC3}
Let $G$ be a graph, $k \geq 2$ and $s \geq 1$ be integers, and let
$\Omega$ be a PMC in $G$ with $|\Omega| > (k - 1)s^3$.  Then there
exists in $G$ either a clique of size $k$ or a minimal separator of
size larger than $s$.
\end{lemma}

\begin{proof}
  By Lemma~\ref{lem:PMC1}, we may assume that for every
  $D \in \cc(G\sm \pmc)$ we have $|N(D)| \leq s$.

  Assume first that for every $x \in \pmc$ the set of non-neighbors of
  $x$ in $\pmc$ (i.e., $\Omega \setminus N[x]$) is of size less than $s^3$.
  Let $A_0 = \pmc$ and consider the following iterative
  process. Given $A_i$ for $i \geq 0$, pick $x_i \in A_i$, and set
  $A_{i+1} = A_i \cap N(x_i)$.  The process terminates when $A_i$
  becomes empty. Clearly, the vertices $x_0,x_1,\ldots$ induce a
  clique. Furthermore, by our assumption,
  $|A_i \setminus A_{i+1}| \leq s^3$.  Therefore this process
  continues for at least $k$ steps, giving a clique of size $k$ in
  $G$.

Thus we are left with the case when there exists $x \in \pmc$ with the set $\pmc \setminus N[x]$ of size at least $s^3$. 
Let $Y = \{x\} \cup (\pmc \setminus N[x])$; we have $|Y| > s^3$, $Y \subseteq \pmc$, and $G[Y]$ is disconnected.

Consider the following iterative process. 
At step $i$, we will maintain a partition $\mathcal{A}_i$ of $Y$ into at least two parts and for every $A \in \mathcal{A}_i$ a set $\mathcal{D}_i(A) \subseteq \cc(G\sm \pmc)$ with the following property:
the sets $\{A \cup \bigcup_{D \in \mathcal{D}_i(A)} D~|~A \in \mathcal{A}_i\}$ is the partition of $G[Y \cup \bigcup_{A \in \mathcal{A}_i} \bigcup_{D \in \mathcal{D}_i(A)} D]$
into vertex sets of connected components. In particular, for every $A \in \mathcal{A}_i$ and $D \in \mathcal{D}_i(A)$ we have $N(D) \cap Y \subseteq A$. 
We start with $\mathcal{A}_0 = \cc(G[Y])$ and $\mathcal{D}_0(A) = \emptyset$ for every $A \in \mathcal{A}_0$.

The process terminates when there exists $A \in \mathcal{A}_i$ of size
larger than $s^2$. Otherwise, we perform a step as follows.  Pick two
distinct $A,B \in \mathcal{A}_i$ and vertices $a \in A$, $b \in B$. By
the properties of $\mathcal{A}_i$, $ab \notin E(G)$.  By
Theorem~\ref{th:PMC2}, there exists $D \in \cc(G\sm \pmc)$ with
$a,b \in N(D)$.  Let
$\mathcal{A} = \{C \in \mathcal{A}_i~|~N(D) \cap C \neq \emptyset\}$.
Note that $A,B \in \mathcal{A}$. Furthermore, since $|N(D)| \leq s$,
we have $2 \leq |\mathcal{A}| \leq s$.

We define
$\mathcal{A}_{i+1} = (\mathcal{A}_i \setminus \mathcal{A}) \cup
\{\bigcup_{C \in \mathcal{A}} C\}$.  For every
$C \in \mathcal{A}_{i+1} \cap \mathcal{A}_i$ we keep
$\mathcal{D}_{i+1}(C) = \mathcal{D}_i(C)$.  Furthermore, we set
$\mathcal{D}_{i+1}(\bigcup_{C \in \mathcal{A}} C) = \{D\} \cup
\bigcup_{C \in \mathcal{A}} \mathcal{D}_i(C)$.  It is straightforward
to verify the invariant for $\mathcal{A}_{i+1}$ and
$\mathcal{D}_{i+1}$.

Furthermore, since every set $C \in \mathcal{A}_i$ is of size at most
$s^2$ while $|Y| > s^3$ we have that $|\mathcal{A}_i| > s$. Since
$2 \leq |\mathcal{A}| \leq s$, we have
$2 \leq |\mathcal{A}_{i+1}| < |\mathcal{A}_i|$. Consequently, the
process terminates after a finite number of steps with $\mathcal{A}_i$
of size at least $2$, $\mathcal{D}_i$, and some $A \in \mathcal{A}_i$
of size greater than $s^2$.

Let $X = A \cup \bigcup_{D \in \mathcal{D}_i(A)} D$ and let
$y \in Y \setminus A$. Note that $G[X]$ is connected by the invariant
on $\mathcal{A}_i$ and $\mathcal{D}_i$, $y$ exists as
$|\mathcal{A}_i| \geq 2$, and $y$ is anticomplete to $X$.  We use
Theorem~\ref{th:PMC2}: for every $a \in A$ fix a component
$D_a \in \cc(G\sm \pmc)$ covering the nonedge $ya$.  Since
$|N(D_a)| \leq s$ while $|A| > s^2$, the set
$\mathcal{D} = \{D_a~|~a \in A\}$ is of size greater than $s$.  Since
$G[X]$ is connected and $y$ is anticomplete to $X$, there exists a
minimal separator $S$ with $y$ in one full side and $X$ in the other
full side.  However, then $S \cap D \neq \emptyset$ for every
$D \in \mathcal{D}$. Hence, $|S| \geq |\mathcal{D}| > s$. This
finishes the proof of the lemma.
\end{proof}

\begin{myproof}[Proof of Theorem~\ref{th:PMC-Pilipczuk-thm}]
  Let $G$ be a graph such that it does not contain a clique on
  $k$ vertices and a minimal separator of size larger than~$s$.
  Let $F$ be a minimal chordal completion of $G$. By
  Lemma~\ref{lem:PMC3}, every maximal clique of $G + F$ is of size at
  most $(k - 1)s^3$. Therefore a clique tree of $G + F$ is a tree
  decomposition of $G$ of width at most $(k - 1)s^3 - 1$, as desired.
\end{myproof}

\section{Nested $2$-wheels}
\label{sec:technical-lemma}

Let $k \geq 0$ be an integer. A {\em $k$-wheel} is a graph formed by a
hole $H$ called the {\em rim} together with a set $C$ of $k$ vertices
that are not in $V(H)$ called the {\em centers}, such that each center
has at least three neighbors in the rim.  We denote such a $k$-wheel
by $(H, C)$.  Observe that a $0$-wheel is a hole.  A $1$-wheel is
called a \emph{wheel} (see Fig. \ref{fig:tc}).  We often write
$(H, u)$ instead of $(H, \{u\})$.

A $2$-wheel $(H, \{u, v\})$ is {\em nested} if $H$ contains two
vertices $a$ and $b$ such that all neighbors of $u$ in $H$ are in one
path of $H$ from $a$ to $b$, while all the neighbors of $v$ are in the
other path of $H$ from $a$ to $b$.  Observe that $a$ and $b$ may be
adjacent to both $u$ and $v$.  As we will see in this section, the
properties of $2$-wheels highlight structural similarities between
(theta, triangle)-free graphs and (even hole, pyramid)-free graphs, in
the sense that in both classes, apart from few exceptions, every
$2$-wheel with non-adjacent centers is nested.

For a center $u$ of a $k$-wheel $(H, C)$, a {\em $u$-sector} of $H$ is
a subpath of $H$ of length at least~1 whose ends are adjacent to $u$
and whose internal vertices are not.  However, a $u$-sector may
contain internal vertices that are adjacent to $v$ for some center
$v\neq u$.  Observe that for every center $u$, the rim of a wheel is
edgewise partitioned into its $u$-sectors.

\subsection*{In (theta, triangle)-free graphs} 

The {\em cube} is the graph formed from a hole of length~6, say
$h_1h_2 \cdots h_6h_1$ together with a vertex $u$ adjacent to $h_1$,
$h_3$, $h_5$ and a vertex $v$ non-adjacent to $u$ and adjacent to
$h_2$, $h_4$, $h_6$.  Note that the cube is a non-nested 2-wheel with
non-adjacent centers.

\begin{lemma}
  \label{lem:(theta, triangle)-free-2wheel-non-adjacent}
  Let $G$ be a (theta, triangle)-free graph. If $W = (H, \{u, v\})$ is
  a 2-wheel in $G$ such that $uv \notin E(G)$, then $W$ is either a
  nested wheel or the cube. 
\end{lemma}

\begin{proof}
  \setcounter{claim}{0}
  Suppose that $W$ is not a nested wheel. We will prove that $W$ is
  the cube. 

 \begin{claim}
  \label{cl:atleast2}
  Every $u$-sector of $H$ contains at most one neighbor of $v$
  and every $v$-sector of $H$ contains at most one neighbor of $u$.
 \end{claim}

 \bpc For otherwise, without loss of generality, some $u$-sector $P= x\dots y$ of $H$
 contains at least two neighbors of $v$. Let $x', y'$ be neighbors of
 $v$ closest to $x, y$ respectively along $P$.  Note that
 $x'y' \notin E(G)$ because $G$ is triangle-free. Since $W$ is not
 nested, $H \sm P$ contains some neighbors of $v$. Note also that
 $H \sm P$ contains some neighbors of $u$.
 
 So, let $Q = z\dots z'$ be the path of $H\sm P$ that is minimal length
 and such that $uz\in E(G)$ and
 $vz'\in E(G)$.  Note that $z'$ is adjacent to either $x$ or $y$, for
 otherwise $uzQz'v$, $uxPx'v$, and $uyPy'v$ form a theta from $u$ to $v$.  So
 suppose up to symmetry that $z'$ is adjacent to $y$. So, $v$ is not
 adjacent to $y$ since $G$ is triangle-free.  It then follows that the
 three paths $vz'y$, $vy'Py$, and $vx'Pxuy$ form a theta, a
 contradiction. \epc
 
 \medskip
 
 \begin{claim}
  \label{cl:noCommon}
  $u$ and $v$ have no common neighbors in $H$.
\end{claim}

\bpc Otherwise, let $x$ be such a common neighbor.  Consider a subpath
$x\dots y$ of $H$ of maximum length with the property of being a
$u$-sector or a $v$-sector, and suppose up to symmetry that it is a
$u$-sector. By its maximality, it contains a neighbor of $v$ different
from $x$. So in total it contains at least two neighbors of $v$, a
contradiction to Claim~\ref{cl:atleast2}.  \epc

 \medskip

 Claim~\ref{cl:atleast2} and~\ref{cl:noCommon} prove that
 $|N_H(u)|=|N_H(v)|$ and the neighbors of $u$ and $v$ alternate along
 $H$.  So, let $x,y,z \in N_H(u)$ and $x',y',z' \in N_H(v)$ be
 distinct vertices in $H$ with
 $x$, $x'$, $y$, $y'$, $z$, $z'$ appearing  in this order along~$H$.  
 If $V(H) =\{x, y, z, x', y', z'\}$, then
 $V(H)\cup \{u, v\}$ induces the cube, so suppose
 $\{x, y, z, x', y', z'\}\subsetneq V(H)$.  Hence, up to symmetry, we
 may assume that $x$, $x'$, $y$, $y'$, $z$ and $z'$ are chosen such
 that: $xz' \notin E(G)$.  But then the
 three paths $vz'(H \sm x)z$, $vy'(H \sm y)z$, and $vx'(H \sm y)xuz$ form a theta, a
 contradiction.
\end{proof}

The following lemma of Radovanovi\'c and Vu\v skovi\'c shows that the
presence of the cube in a (theta, triangle)-free graph entails some
structure.

\begin{lemma}[see \cite{radovanovicV:theta}]
  \label{lem:contain-cube}
  Let $G$ be a (theta, triangle)-free graph. If $G$ contains the cube,
  then either it is the cube, or it has a clique separator of size at
  most~2.
\end{lemma}

\subsection*{In even-hole-free graphs}

Let us consider a classical generalization of even-hole-free graphs.  

A {\em prism} is a graph made of three vertex-disjoint chordless paths
$P_1 = a_1 \dots b_1$, $P_2 = a_2 \dots b_2$, $P_3 = a_3 \dots b_3$ of
length at least 1, such that $a_1a_2a_3$ and $b_1b_2b_3$ are triangles
and no edges exist between the paths except those of the two
triangles (see Fig.~\ref{fig:tc}). 
An \emph{even wheel} is a wheel $(H, u)$ such that $u$ has an even number of
neighbors in $H$.  A \emph{square} is a hole of length~4. 

It is easy to see that all thetas, prisms, even wheels, and squares contain even
holes.  The class of (theta, prism, even wheel, square)-free
graphs is therefore a generalization of even-hole-free graphs that
capture the structural properties that we need here. 

A proof of the following lemma can be found
in~\cite{chudetal:maxStEHFPyramfree} (where it relies on many
lemmas). We include here our self-contained proof for the sake of
completeness.
Call a wheel \emph{proper} if it is not pyramid.
A {\em cousin wheel} is a 2-wheel made of a hole
$H = h_1 h_2 \dots h_n h_1$ and two non-adjacent centers $u$ and $v$,
such that $N_H(u) = \{h_1,h_2,h_3\}$ and $N_H(v) = \{h_2,h_3,h_4\}$.

\begin{lemma}
  \label{lem:EHF-2wheel-non-adjacent-1}
  Let $G$ be a (theta, prism, pyramid, even wheel, square)-free graph. If
  $W = (H, \{u, v\})$ is a 2-wheel in $G$ such that $uv \notin E(G)$,
  then $W$ is either a nested or a cousin wheel.  Moreover, if $W$ is
  nested then $|N_H(u) \cap N_H(v)| \leq 1$.
\end{lemma}

\begin{proof}  
 \setcounter{claim}{0} 
 
 In the case where $W = (H, \{u, v\})$ is nested, it must be that 
 $|N_H(u) \cap N_H(v)| \leq 1$, for otherwise $G$ would contain a square. 
 Since $G$ contains no even wheel, it is sufficient
 to consider the following cases.
    
  \medskip
  \noindent \textbf{Case 1:} $N_H(u) =3$ or $N_H(v)=3$.

  \vspace{.5ex}
  Assume that $W$ is not a nested
  wheel.  We will prove that $W$ is a cousin wheel.
  Without loss of generality, we may assume that $|N_H(u)| = 3$, and let
  $N_H(u) = \{x, y, z\}$.  We denote by $P_x=y\dots z$, $P_y=x\dots z$
  and $P_z=x\dots y$ the three $u$-sectors of $H$.  

  Suppose $xyz$ is a path of $H$. Then $v$ must be adjacent to $y$, for otherwise
  $W$ is nested, a contradiction. Since $V(H) \cup \{u\}$ and $V(H \sm y) \cup \{u,v\}$
  do not induce an even wheel, $v$ has exactly two neighbors in $P_y$. Moreover, the two neighbors 
  of $v$ in~$P_y$ are adjacent, for otherwise $H\sm y$, $u$, and $v$ form
  a theta. Since $(H,v)$ is not a pyramid, this means that one of $x$ or $z$
  is a neighbor of~$v$. Therefore, $W$ is a cousin wheel.

  Now suppose that $\{x, y, z\}$ does not
  induce a path. So $xy$, $yz$, and $zx$ are non-edges.
  Note that $v$ is adjacent to at most one of $x$, $y$, or $z$, because
  $G$ contains no square. Up to symmetry, assume that $vx\notin E(G)$. 
  Let $R$ be the $v$-sector of $H$ which contains $x$ (in its interior). 
  Since $(H, \{u, v\})$ is not a nested wheel, the ends of $R$ are not both in
  $P_x$, or both in $P_y$, or both in $P_z$.
  So assume that $R=y'\dots z'$
  with $z'$ is in the interior of $P_z$ and $y'$ is not in $P_z$.  If
  $y'$ is in $P_x$, then $R$, $u$, and $v$ form a theta from $x$ to $z$, a contradiction.
  Hence, $y'$ is not in $P_x$, so $y'$ is in the interior of $P_y$.
  
  Call $x'$ the neighbor of $v$ in $H$ different from $y'$ and $z'$.
  If $x'$ is not in the interior of $P_x$, then $P_x$ is
  contained in the $v$-sector $x'Hz'$. Thus, there exists a $v$-sector $S$
  which contains $P_x$. In particular, the hole made of $S$ and $v$ contains
  two non adjacent neighbors of~$u$, namely $y$ and $z$. Hence, $S$, $u$, and $v$
  form a theta from $y$ to $z$.
  So, $x'$ is in the interior of $P_x$.
  
  This means $x$, $y'$, $z$, $x'$, $y$, $z'$ appear in this order along $H$. 
  If $x'z \notin E(G)$, then the paths $x'(H \sm y)z$, $x'(H \sm z)yuz$, and
  $x'vy'(H \sm x)z$ form a theta from $x'$ to $z$, a contradiction.
  So, $x'z \in E(G)$. By symmetry, $x'y \in E(G)$. But then,
  $\{u,y,x'z\}$ induces a square, a contradiction. 
  
  \medskip
  \noindent \textbf{Case 2:} $N_H(u)\geq 5$ and $N_H(v) \geq 5$
 
  \medskip
  
  For a contradiction, suppose that $(H, \{u,v\})$ is not a nested wheel. 
  First of all, we have $N_H(u) \neq N_H(v)$, for otherwise $u$, $v$, and two non-adjacent vertices of
  $N_H(u)$ would form a square.
  So in $H$, there exists 
  a neighbor of $v$ that is not adjacent to $u$. It is therefore well
  defined to consider the $u$-sector $P = x\dots y$ of $H$ whose
  interior contains $k\geq 1$ neighbors of $v$, and to choose such a
  sector with $k$ minimum.
  We denote by $x'$ the neighbor of $x$ in $H\sm P$, by $y'$ the
  neighbor of $y$ in $H\sm P$ and by $Q = x' \dots y'$ the path $H\sm P$.
 
  Note that $u$ has some neighbor in the interior of $Q$,
  because $u$ has at least~5 neighbors in $H$.  We now show that $v$
  also has some neighbor in the interior of $Q$.  Suppose that it is
  not the case.  Then, the neighborhood of $v$ in~$H$ is completely
  contained in $V(P) \cup \{x',y'\}$.  Since $(H, \{u, v\})$ is not a
  nested wheel, $v$ is adjacent to $x'$ or $y'$ --- and in fact to
  both of them, for otherwise the hole $uxPyu$ would contain an even
  number (at least~4) of neighbors of~$v$, thus inducing an even
  wheel, a contradiction.  Now since $\{u, v, x, y\}$ does not induce
  a square, up to symmetry we may assume that $vx\notin E(G)$. Since
  $|N_H(v)|\geq 5$, $v$ has at least 2 neighbors in the interior of
  $P$, and so $k\geq 2$.  Note that $u$ is adjacent to~$x'$, for
  otherwise, $x'$ would be the unique neighbor of $v$ in the interior
  of a $u$-sector, contradicting the minimality of $k$. Since
  $\{u, v, x', y'\}$ does not induce a square, we know that $u$ is not
  adjacent to~$y'$. But then, $y'$ is the unique neighbor of $v$ in
  the interior of some $u$ sector, a contradiction to the minimality
  of $k$.  This proves that $v$ has some neighbor in the interior of
  $Q$. 
   

%

 By the fact that each of $u$ and $v$ has some neighbor in the interior of $Q$, 
 a path $S$ from $u$ to $v$ whose interior is in the interior of $Q$ exists.  Let $x''$
 (resp.\ $y''$) be the neighbor of $v$ in $P$ closest to $x$ (resp.\ $y$) along $P$.  
 If $x''=y''$, then $x''$ is an internal vertex of $P$, and so $S$ and $P$ form a
 theta from $u$ to $x''$.  If $x''y''\in E(G)$, then $S$ and $P$ form a
 pyramid.  If $x''\neq y''$ and $x''y''\notin E(G)$, then $S$,
 $uxPx''v$, and $uyPy''v$ form a theta from $u$ to $v$. Each of the cases
 yields a contradiction; this completes the proof.
\end{proof}

\section{Bounding the treewidth}
\label{sec:bounding-tw}

In this section, we prove that the treewidth is bounded in (theta,
triangle, $S_{i, j, k}$)-free graphs and in (even hole, pyramid,
$K_t$, $S_{i, j, k}$)-free graphs.

For (theta, triangle)-free graphs, by
Lemma~\ref{lem:contain-cube}, we may assume that the graphs we work on
are cube-free since the cube itself has small treewidth, and clique
separators of size at most~2 in some sense preserve the treewidth
(this will be formalized in the proofs).  For (even hole, pyramid)-free
graphs, recall that we work from the start in a superclass, namely
(theta, prism, pyramid, even wheel, square)-free graphs.

Since our proof is the same for (theta, triangle, $S_{i, j, k}$)-free
graphs and (even hole, pyramid, $K_t$, $S_{i, j, k}$)-free graphs, to
avoid duplicating it, we introduce a class $\cal C$ that contains all
the graphs that we need to consider while entailing the structural
properties that we need.

Call \emph{butterfly} a wheel $(H, v)$ such that
$N_H(v) = \{a, b, c, d\}$ with $ab\in E(G)$, $bc\notin E(G)$,
$cd\in E(G)$ and $da\notin E(G)$.  Let $\cal C$ be the class of all
(theta, prism, pyramid, butterfly)-free graphs such that every 2-wheel
with non-adjacent centers is either a nested or a cousin wheel.

\begin{lemma}
  \label{l:inC}
  If $G$ is a (theta, triangle, cube)-free graph or a (theta, prism,
  pyramid, even wheel, square)-free graph, then $G\in \cal C$.
\end{lemma}

\begin{proof}
  If $G$ is a (theta, triangle, cube)-free graph, then $G$ is
  theta-free and (prism, pyramid, butterfly)-free (because prisms,
  pyramids, and butterflies contain triangles).  Furthermore, every
  2-wheel with non-adjacent centers is a nested wheel 
  by Lemma~\ref{lem:(theta, triangle)-free-2wheel-non-adjacent}.

  If $G$ is a (theta, prism, pyramid, even wheel, square)-free graph,
  then $G$ is (theta, prism, pyramid)-free and butterfly-free (because
  a butterfly is an even wheel).  Furthermore, every 2-wheel with
  non-adjacent centers is either a nested or a cousin wheel by
  Lemma~\ref{lem:EHF-2wheel-non-adjacent-1}.

  Hence $G\in \cal C$ as claimed.
\end{proof}

For our proof, we need a special kind of $k$-wheel. A {\em
  $k$-span-wheel} is a $k$-wheel $(H,C)$ that satisfies the following
properties.

\begin{itemize}
\item There exist two non-adjacent vertices $x, y$ in $H$ and we denote by
  $P_A = a_1\dots a_\alpha$ and $P_B = b_1 \dots b_\beta$ the two
  paths of $H$ from $x$ to $y$, with $x=a_1=b_1$ and
  $y=a_\alpha = b_\beta$.
\item $C\cup \{x, y\}$ is an independent set. 
\item There exists an ordering of vertices in $C$, namely
  $v_1,v_2,\cdots,v_{k}$.
\item Every vertex of $C$ has neighbors in the interiors of both $P_A$
  and $P_B$ (and at least 3 neighbors in $H$ since $(H, C)$ is a
  $k$-wheel).
\item For every $1\leq i < j \leq k$ and $1\leq i', j' \leq \alpha$,
  if $v_ia_{i'}\in E(G)$ and $v_ja_{j'}\in E(G)$ then $i'\leq j'$.
\item For every $1\leq i < j \leq k$ and $1\leq i', j' \leq \beta$, 
  if $v_ib_{i'}\in E(G)$ and $v_jb_{j'}\in E(G)$ then $i'\leq j'$. 
\end{itemize}

Informally, a $k$-span-wheel is such that, walking from $x$ to $y$
along both $P_A$ and $P_B$, one first meets all the neighbors of
$v_1$, then all neighbors of $v_2$, and so on until $v_k$.  Observe
that a $1$-span-wheel is a wheel, $2$-span-wheel is a nested 2-wheel.
Note that distinct $v_i$ and $v_j$ may share common neighbors on $H$
(it is even possible that
$N_{P_A}(v_1) = \dots = N_{P_A}(v_k) = \{a_i\}$).

Observe that in the following theorem, thetas, pyramids, prisms, and
butterflies have to be excluded, since they do not satisfy the conclusion.

\begin{lemma}
  \label{th:stable-induces-wheel}
  Let $G$ be a graph in $\cal C$. Let $C$ be a minimal separator in
  $G$ of size at least~2 that is furthermore an independent set, and $A$ and
  $B$ be components of $G\sm C$ that are full to $C$.  Then:

  \begin{enumerate}
  \item\label{i:stuffOne} There exist two vertices $x$ and $y$ in $C$,
    a path $P_A$ from $x$ to $y$ with interior in $A$, and a path
    $P_B$ from $x$ to $y$ with interior in $B$ such that all vertices
    in $C\sm \{x, y\}$ have neighbors in the interior of both $P_A$
    and $P_B$.  Note that $V(P_A) \cup V(P_B)$ induces a hole that we
    denote by $H$.
    
  \item\label{i:stuffTwo} $(H, C\sm \{x, y\})$ is a
    $(|C|-2)$-span-wheel.
  \end{enumerate}
\end{lemma}

\begin{proof}
  \setcounter{claim}{0} 

  We first prove~\ref{i:stuffOne}, by induction on $k=|C|$.

  If $k=2$, then $x$, $y$, $P_A$, and $P_B$ exist from the connectivity
  of $A$ and $B$, and the conditions on $C\sm \{x, y\}$ vacuously
  hold.  So suppose the result holds for some $k\geq 2$, and let us
  prove it for $k+1$.  Let $z$ be any vertex from $C$, and apply the
  induction hypothesis to $C\sm z$ in $G\sm z$.  This provides two
  vertices $x, y$ in $C\sm z$ and two paths $P_A$ and $P_B$.  We
  denote by $H$ the hole formed by $P_A$ and~$P_B$.

  \begin{claim}
    \label{c:indHyp}
    Every vertex in $C\sm \{x, y, z\}$ has neighbors in the interior
    of both $P_A$ and $P_B$.
  \end{claim}

  \bpc
  Follows directly from the induction hypothesis. 
  \epc \medskip
  
  Since $z$ has a neighbor in $A$ and $A$ is connected, there exists a
  path $Q_A = z\dots z_A$ in $A\cup \{z\}$, such that $z_A$ has a
  neighbor in the interior of $P_A$.  A similar path $Q_B$ exists.  We
  set $Q = z_AQ_A z Q_B z_B$.  We suppose that $x$, $y$, $P_A$, $P_B$,
  $Q_A$, and $Q_B$ are chosen subject to the minimality of $Q$.
  
  Observe that $Q$ is a chordless path by its minimality and the fact that $A$
  and $B$ being anticomplete. The minimality of $Q$ implies that the
  interior of $Q$ is anticomplete to the interior of $P_A$ and to the
  interior of $P_B$.

  \medskip
  
  \begin{claim}
    \label{cl:Qal1}
    We may assume that $Q$ has length at least~1.  
  \end{claim}

  \bpc Otherwise, $z=z_A=z_B$, so $z$ has neighbors in the interior of
  both $P_A$ and $P_B$.  Hence, by Claim~\ref{c:indHyp}, $x$, $y$,
  $P_A$, and $P_B$ satisfy~\ref{i:stuffOne}.  \epc \medskip

 Let $a$ (resp.\ $a'$) be the neighbor of $z_A$ in $P_A$ closest to
  $x$ (resp.\ to $y$) along $P_A$.  Let $b$ (resp.\ $b'$) be the
  neighbor of $z_B$ in $P_B$ closest to $x$ (resp.\ to $y$) along~$P_B$. 
  
  \medskip
  
  \begin{claim}
    \label{cl:aneqap}
    If $a\neq a'$ and $aa'\notin E(G)$, then $z=z_A$.    
    If $b\neq b'$ and $bb'\notin E(G)$, then $z=z_B$.  
  \end{claim}

  \bpc We give a proof only for the statement of~$a$, since
  the proof for $b$ is similar.
  
  For suppose $a\neq a'$, $aa'\notin E(G)$, and $z\neq z_A$,
  let $z'$ be the neighbor of $z_A$ in $Q$.
  Set $P'_A=xP_Aaz_aa'P_Ay$ and $Q'=z'Qz_B$.  Let us prove that $x$,
  $y$, $P'_A$, $P_B$, and $Q'$ contradict the minimality of $Q$.
  Obviously, $Q'$ is shorter than $Q$, so we only have to prove that
  every vertex in $C\sm \{z\}$ has neighbors in the interior of both
  $P'_A$ and $P_B$.  For $P_B$, it follows from Claim~\ref{c:indHyp}.
  So suppose for a contradiction that a vertex $c\in C\sm \{z\}$ has
  no neighbor in the interior of $P'_A$.  Since by
  Claim~\ref{c:indHyp} $c$ has a neighbor $c'$ in the interior of
  $P_A$, $c'$ is an internal vertex of $aP_Aa'$.  Since $G$ is
  theta-free, $(H, z_A)$ is a wheel.  Note that $(H, \{c, z_A\})$ is
  not nested because of $c'$ and some neighbor of $c$ in the interior
  of $P_B$ (i.e.\ the neighborhood of~$c$ in~$H$ is not contained in a
  unique $z_A$-sector).  
  Since $G \in \mathcal{C}$, by Lemma~\ref{lem:EHF-2wheel-non-adjacent-1},
  $(H, \{c, z_A\})$ is a cousin wheel. 
  Since $c$ has neighbors in the interiors of both $P_A$ and $P_B$,
  this means that $x$ or $y$
  is a common neighbor of $c$ and $z_A$, a contradiction to $C$ being
  an independent set.
  The proof for the latter statement (with $b$) is similar. 
  \epc \medskip

  \begin{claim}
    \label{cl:justy}
    We may assume that $x$ has neighbors in the interior of $Q$ and
    $y$ has no neighbor in the interior of $Q$.
  \end{claim}

  \bpc We show that if it is not the case, then there is a contradiction.
  For suppose both $x$ and $y$ have a neighbor in the interior of~$Q$, then
  a path of minimal length from $x$ to $y$ with interior in the
  interior of $Q$ form a theta together with $P_A$ and $P_B$, a contradiction.  
  
  Now
  suppose that none of $x$ and $y$ has a neighbor in the interior of $Q$.
  Recall that Claim~\ref{cl:Qal1} tells us that $z_A\neq z_B$. So either $z\neq z_A$ or
  $z\neq z_B$. Up to symmetry, we may assume that $z \neq z_A$.
  Hence by Claim~\ref{cl:aneqap}, either $a=a$ or $aa' \in E(G)$.
  
  Suppose $a=a'$. This implies that $a$ is in the interior
  of $P_A$.  If $b=b'$, then $b$ is in the interior of $P_B$ --- so $H$
  and $Q$ form a theta from $a$ to~$b$; if $bb'\in E(G)$, then $H$
  and $Q$ form a pyramid; and if $b\neq b'$ and $bb'\notin E(G)$, then
  $aP_AxP_Bbz_B$, $aP_AyP_Bb'z_B$, $az_AQz_B$ form a theta from $a$ to
  $z_B$ (note that $az_AQz_B$ has length at least~2 because $z_A \neq z_B$), 
  a contradiction. So, $aa'\in E(G)$. 
  
  Suppose that $bb'\in E(G)$.
  Note that $|\{a, a'\}\cap \{b, b'\}\cap \{x, y\}| \neq 2$, 
  because $x$ and $y$ are not adjacent.
  Moreover, $|\{a, a'\}\cap \{b, b'\}\cap \{x, y\}| \neq \emptyset$,
  for otherwise $H$ and $Q$ form a prism.
  So,~$|\{a, a'\}\cap \{b, b'\}\cap \{x, y\}|=1$. In
  this last case, we suppose up to symmetry that $x=a=b$.  So, $z$ is
  in the interior of $Q$ since it is non-adjacent to $x$ --- 
  in particular $Q$ has length at least~2. Hence, $H$ and $Q$ form a
  butterfly (with $x = a = b$ being the center), a contradiction.  
  
  So, $bb'\notin E(G)$. If $b=b'$, then $b$ is in the interior of~$P_B$;
  thus $P_A$, $P_B$, and $Q$ form a pyramid (i.e.\ $3PC(az_Aa', b)$, a contradiction.
  So, $b \neq b'$, and hence by Claim~\ref{cl:aneqap}, $z_B = z$.
  This means that $b \neq x$ and $b' \neq y$ (because $C$ is an independent set).
  Therefore, $aP_AxP_Bbz$, $a'P_AyP_Bb'z$, and $z_AQz$ form a pyramid 
  (i.e.\ $3PC(aa'z_A,z)$), a contradiction. 
  
  So, each case leads to a contradiction.
  Hence, exactly one of $x$ or $y$ has neighbors in the interior of
  $Q$, and up to symmetry we may assume it is $x$. 
 \epc  \medskip

  \begin{claim}
    \label{cl:apx}
     $a'x \in E(G)$ and $b'x\in E(G)$.  
  \end{claim}

  \bpc
  First, suppose $z_A$ is adjacent to~$x$, i.e.\ $a = x$.  Then, $z_A\neq z$ since $C$ is an
  independent set.  Note that $a=a'$ is impossible since $z_A$ has neighbors in
  the interior of~$P_A$. So, by Claim~\ref{cl:aneqap}, $a'x\in
  E(G)$. 
  
  Now suppose $z_A$ is not adjacent to~$x$.  By Claim~\ref{cl:justy}, $x$ has a
  neighbor in the interior of $Q$, so we choose such a neighbor $x'$
  closest to $z_A$ along $Q$. 
  Note that by the minimality of~$Q$, no vertex in the interior of~$Q$
  has neighbor in the interior of~$P_A$ and in the interior of~$P_B$. 
  Since $y$ is not adjacent to~$x'$ (by Claim~\ref{cl:justy}), 
  $x'$ has no neighbors in $(P_A \cup P_B)\sm \{x\}$. 
  We set $R=xx'Qz_A$ and observe that $R$ has length at
  least~2.  If $a\neq a'$ and $aa'\notin E(G)$, then $xP_Aaz_A$,
  $xP_ByP_Aa'z_A$, and $R$ form a theta from $x$ to $z_A$.
  If $aa'\in E(G)$, then $P_A$, $P_B$,
  and $R$ form a pyramid.  Therefore $a=a'$. Note that 
  $xa \in E(G)$, for otherwise, $P_A$, $P_B$, and $R$ 
  form a theta. Hence, $a'x\in E(G)$.
  
  The proof for $b'x \in E(G)$ is similar. 
  \epc  \medskip

  To conclude the proof of~\ref{i:stuffOne}, set $P'_A= zQz_Aa'P_Ay$
  and $P'_B= zQz_Bb'P_By$.  By Claim~\ref{cl:apx}, $x$ has neighbors
  in the interior of both $P'_A$ and $P'_B$ (these neighbors are $a'$
  and $b'$). Note that since $a'x, b'x \in E(G)$, the
  interiors of $P_A$ and $P_B$ are included in the interiors of $P'_A$
  and $P'_B$ respectively.
  Hence, by Claim~\ref{c:indHyp}, every vertex of $C\sm z$
  has neighbors in the interior of both $P'_A$ and $P'_B$.

 Hence, the vertices $z,y$ and the paths
  $P'_A$ and $P'_B$ show that~\ref{i:stuffOne} is satisfied.

  \medskip

  Let us now prove~\ref{i:stuffTwo}.  Note that $(H, C\sm \{x, y\})$ is
  $(|C|-2)$-wheel (this follows because $G$ is theta-free, every vertex in
  $C\sm \{x, y\}$ has at least three neighbors in $H$). It remains to
  prove that it is a $(|C|-2)$-span-wheel.  Note that it is clearly
  true if $|C|\leq 3$.  We set $P_A = a_1\dots a_\alpha$ and
  $P_B = b_1 \dots b_\beta$ with $x=a_1=b_1$ and
  $y=a_\alpha = b_\beta$, as in the definition of a $k$-span-wheel. We
  just have to exhibit an ordering of the vertices of $C\sm\{x, y\}$
  that satisfies the rest of the definition.

  We first define $v_1$, the smallest vertex in the order we aim to
  construct. Note that no vertex $v\in C\sm \{x, y\}$ is adjacent to $x$ or $y$,
  because $C$ is an independent set. We let $v_1$ be a vertex of $C$
  that is adjacent to $a_i$ with $i$ minimum. Let $j$ be the smallest
  integer such that $v_1$ is adjacent to $b_j$.  We suppose that $v_1$
  is chosen subject to the minimality of $j$.  Let $i', j'$ be the
  greatest integers such that $v_1$ is adjacent to $a_{i'}$ and
  $b_{j'}$.  Note that $1 < i \leq i' < \alpha$  and $1 < j \leq j' < \beta$.

  \begin{claim}
    \label{cl:minv1}
    For every $w\in C\sm \{x, y, v_1\}$, we have
    $N_H(w) \subseteq V(a_{i'}P_A y P_Bb_{j'})$.
  \end{claim}

  \bpc We first note that the $2$-wheel $(H, \{v_1, w\})$ is not a
  cousin wheel, because this may happen only when $x \in N({v_1})$ or $y \in N({v_1})$
  (recall that if it was a cousin wheel, $N_H(v_1)$ would induce a 3-vertex path in $H$).

  Hence, $(H, \{v_1, w\})$ is a nested wheel. Suppose that
  $N_H(w) \not\subseteq V(a_{i'}P_Aa_\alpha) \cup V(b_{j'}P_Bb_\beta)$.
  This means that $w$ has a neighbor $z$ in $a_{i'-1}P_AxP_Bb_{j'-1}$.
  Since $(H, \{v_1, w\})$ is a nested wheel, $N_H(w)$ is contained
  in a $v_1$-sector $Q$ of $(H, v_1)$.  Moreover, since $w$ has a neighbor in the
  interior of both $P_A$ and $P_B$, we have $Q=a_{i}P_AxP_Bb_{j}$.
  Since $H$ and $w$ form a wheel, $w$ has neighbor in the interior of $Q$.
  This contradicts the minimality of $i$ or $j$. 
  \epc \medskip 

  The order of $C\sm \{x, y\}$ is now constructed as follows: we
  remove $v_1$ from $C$, define $v_2$ as we defined $v_1$
  (minimizing $i$, and then minimizing $j$), then remove $v_2$, define
  $v_3$, and so on. This iteratively constructs an ordering of
  $C\sm \{x, y\}$ showing that $(H, C\sm \{x, y\})$ is a
  $(|C|-2)$-span-wheel.
\end{proof}

For integers $t, k \geq 1$, the Ramsey number $R(t,k)$ is the smallest
integer $n$ such that any graph on $n$ vertices contains either a
clique of size $t$, or an independent set of size $k$.

\begin{theorem}
  \label{th:mainbound}
  An ($l$-span-wheel, $K_t$)-free graph $G\in \cal C$ has treewidth at
  most $(t-1) (R(t, l+2)-1)^3 - 1$.
\end{theorem}

\begin{proof}
  Suppose for a contradiction that the treewidth of $G$ is at least
  \linebreak $(t-1) (R(t, l+2)-1)^3$.  Since $G$ is $K_t$-free, by
  Theorem~\ref{th:PMC-Pilipczuk-thm} $G$ admits a minimal separator
  $D$ of size at least $R(t,l+2)$.  Let $A$ and $B$ be two connected
  components of $G\sm D$ that are full to $D$.  By the definition of Ramsey number,
  $G[D]$ contains an independent set $C$ of size $l+2$.  We define
  $G'= G[A\cup C \cup B]$, and observe that $C$ is a minimal separator
  of $G'$.  Hence by Lemma~\ref{th:stable-induces-wheel} applied to
  $G'$, the graph contains an $l$-span-wheel, a contradiction.
\end{proof}

The following shows that in $\cal C$, an $l$-span-wheel with large $l$
contains $S_{i, j, k}$ with large $i, j, k$.

\begin{lemma}
  \label{l:4k}
  If a butterfly-free graph $G$ contains a $(4k+1)$-span-wheel with
  $k \geq 0$, then it contains $S_{k+1, k+1, k+1}$.
\end{lemma}

\begin{proof}
  \setcounter{claim}{0} 
  Consider a $(4k+1)$-span-wheel in~$G$, with
  $x$, $y$, $P_A$, and $P_B$ be as in the definition of
  span-wheel given in the beginning of the current section.
  Let $v_1, \dots, v_{4k+1}$ be the centers of
  the span wheel.  For each $i=1, \dots, 4k+1$, let $a_i$ (resp.\
  $a'_i$) be the neighbor of $v_i$ in $P_A$ closest to $x$ (resp.\ to
  $y$) along $P_A$.  Let $b_i$ (resp.\ $b'_i$) be the neighbor of
  $v_i$ in $P_B$ closest to $x$ (resp.\ to $y$) along $P_B$.  We set
  $P_i=a_i P_A x P_B b_i$ and $Q_i=a'_i P_A y P_B b'_i$.

\begin{claim}
  \label{cl:lengthP}
  $P_i$ has length at least $i+1$ and $Q_i$ has length at least
  $4k+3-i$.
\end{claim}
 
\bpc We prove this by induction on $i$ for $P_i$.  It is clear that
$P_1$ has length at least 2 since $x$ is not adjacent to $v_1$.
Suppose the claim holds for some fixed $i\geq 1$, and let us prove it
for $i+1$.  From the induction hypothesis, $P_i$ has length at least
$i+1$, and since $v_i$ has a neighbor in the interior of $P_{i+1}$
(because it has at least three neighbors in $H$), the length of $P_{i+1}$ is
greater than the length of $P_i$, so $P_{i+1}$ has length at least
$i+2$.

The proof for $Q_i$ is similar, except we start by proving that
$Q_{4k+1}$ has length at least~2, and that the induction goes backward
down to $Q_1$.  \epc \medskip

We set $l=2k+1$. So, by Claim~\ref{cl:lengthP}, $P_l$ and $Q_l$ both
have length at least~$2k+2$.  We set $v=v_l$, $P=P_l$, $Q=Q_l$,
$a=a_l$, $a'=a'_l$, $b=b_l$ and $b'=b'_l$.
Since $G$ is butterfly-free, we do not have $aa'\in E(G)$ and
$bb'\in E(G)$ simultaneously.  
So, up to symmetry we may assume that either $a=a'$; or
$a\neq a'$ and $aa'\notin E(G)$.

  If $a=a'$, let $u$, $u'$, and $u''$  be three distinct vertices in $P$
  such that $a$, $u$, $u'$, and $u''$ appear in this order along $P$,
  $aPu$ has length $k+1$ and $bPu''$ has length  $k-1$ (which is possible 
  because $P$ has length at least $2k+1$).  Let $w$ be in
  $Q$ and such that $aQw$ has length $k+1$ (which is possible 
  because $Q$ has length at least $2k+1$).  The three paths $aPu$,
  $avbPu''$, and $aQw$ form an $S_{k+1, k+1, k+1}$. 

  If $a\neq a'$ and $aa'\notin E(G)$, then let $u$, $u'$, and $u''$ be
  three distinct vertices in $P$ such that $a$, $u$, $u'$, and $u''$
  appear in this order along $P$, $aPu$ has length $k$ and $bPu''$ has
  length $k$.  Let $w$ be in $Q$ and such that $a'Qw$ has length
  $k$.  The three paths $vaPu$, $vbPu''$ and $va'Qw$ form an
  $S_{k+1, k+1, k+1}$.
\end{proof}

\medskip

The following is a classical result on treewidth and we omit its
proof.

\begin{lemma}[\cite{lovasz:minor}]
  \label{lem:treewidth-atoms}
  The treewidth of a graph $G$ is the maximum treewidth of an induced
  subgraph of $G$ that has no clique separator. 
\end{lemma}

\begin{theorem}
  \label{th:ttf}
  For $k\geq 1$, a (theta, triangle, $S_{k, k, k}$)-free graph $G$ has
  treewidth at most $2(R(3, 4k-1))^3-1$.
\end{theorem}

\begin{proof}
  By Lemma~\ref{lem:treewidth-atoms}, it is enough to consider a graph $G$
  that does not have a clique separator.
  If $G$ contains a cube, then Lemma~\ref{lem:contain-cube}
  tells us that $G$ itself is the cube. 
  By classical results on treewidth, the treewidth of the cube is 3
  (but the trivial bound 8 would be enough for our purpose), which 
  in particular achieves the given bound.  
  We may therefore assume that $G$ is cube-free.
  Moreover, by Lemma~\ref{l:inC}, $G$ is in
  $\mathcal{C}$.  Since $G$ is $S_{k,k,k}$-free, by Lemma~\ref{l:4k}, 
  $G$ contains no $(4k-3)$-span-wheel. Moreover, $G$ contains no
  $K_3$ by assumption.  Hence, by Theorem~\ref{th:mainbound}, $G$ has
  treewidth at most $2(R(3, 4k-1))^3-1$.
\end{proof}

\begin{theorem}
  \label{th:ehf}
  For $k\geq 1$, an (even hole, pyramid,  $K_t$, $S_{k, k, k}$)-free graph $G$ has
  treewidth at most $(t-1)(R(t, 4k-1))^3-1$.
\end{theorem}

\begin{proof}
  Since all thetas, prisms, even wheels, and squares contain even holes,
  $G$ is (theta, prism,
  pyramid, even wheel, square)-free.  So, 
  by Lemma~\ref{l:inC}, $G$ is in~$\mathcal{C}$.  Since $G$ is $S_{k,k,k}$-free,
  by Lemma~\ref{l:4k}, $G$
  contains no $(4k-3)$-span-wheel. Moreover, $G$ contains no $K_t$ by assumption.  Hence,
  by Theorem~\ref{th:mainbound}, $G$ has treewidth at most
  $(t-1)(R(t, 4k-1))^3-1$.
\end{proof}

\section*{Acknowledgement}

Thanks to \'Edouard Bonnet, Zden\v ek Dvo\v r\'ak, Kristina V\v
uskovi\'c, and R\'emi Watrigant for useful discussions.
We are also grateful to two anonymous referees 
for careful reading of the paper and valuable suggestions and comments,
which have improved the presentation of the paper.


\end{document}